\documentclass[11pt, letterpaper]{article}

\usepackage[letterpaper,left=1in,right=1in,top=1in,bottom=1in]{geometry}
\usepackage{graphicx}
\usepackage{amsmath}
\usepackage{amsthm}
\usepackage{amssymb}
\usepackage{setspace}
\usepackage{epstopdf}
\usepackage{color}
\usepackage{booktabs}
\usepackage{multirow}
\usepackage{subfigure}
\usepackage[ruled,linesnumbered]{algorithm2e}
\usepackage[normalem]{ulem}

\newtheorem{mydef}{Definition}
\newtheorem{myass}{Assumption}
\newtheorem{myth}{Theorem}

\newtheorem{myremark}{Remark}

\newcommand{\m}{\mathbb}
\newcommand{\s}{\mathcal}
\newcommand*{\dif}{\mathop{}\!\mathrm{d}}

\SetKwBlock{Repeat}{Repeat $n_{update}  \text{ times}$}{end}

\title{\Large \bf Lyapunov-based reinforcement learning for distributed control with stability guarantee}

\author{
\centerline{\normalsize Jingshi Yao$^{a}$, Minghao Han$^{b}$, Xunyuan Yin$^{a,b,}$\thanks{Corresponding author: Xunyuan Yin. Tel: (+65) 6316 8746. Email: xunyuan.yin@ntu.edu.sg}}
\vspace{5mm}\\
\centerline{\small $^{a}$School of Chemistry, Chemical Engineering and Biotechnology, Nanyang Technological University,}\\
\centerline{\small 62 Nanyang Drive, 637459, Singapore}\\
\centerline{\small $^{b}$ Nanyang Environment and Water Research Institute (NEWRI),}\\
\centerline{\small Nanyang Technological University, 1 CleanTech Loop, 637141, Singapore}\\
}
\allowdisplaybreaks
\date{}

\begin{document}

\maketitle
\setstretch{1.5}

\begin{abstract}

In this paper, we propose a Lyapunov-based reinforcement learning method for distributed control of nonlinear systems comprising interacting subsystems with guaranteed closed-loop stability. Specifically, we conduct a detailed stability analysis and derive sufficient conditions that ensure closed-loop stability under a model-free distributed control scheme based on the Lyapunov theorem. The Lyapunov-based conditions are leveraged to guide the design of local reinforcement learning control policies for each subsystem. The local controllers only exchange scalar-valued information during the training phase, yet they do not need to communicate once the training is completed and the controllers are implemented online. The effectiveness and performance of the proposed method are evaluated using a benchmark chemical process that contains two reactors and one separator.

\end{abstract}

\section{Introduction}
As modern industrial manufacturing advances to meet growing demands, the complexity and scale of chemical manufacturing processes have increased rapidly \cite{daoutidis2016JPC,CACE2023,yin2017automatica,TCST2019,burg2016large}. In response, the distributed control architecture has emerged as a scalable and flexible framework for advanced process control in large-scale industrial operations \cite{daoutidis2018cace,christofides2013distributed,aiche2019}. Within this framework, a large-scale industrial plant is typically partitioned into multiple interacting subsystems, each of which is managed by a local controller. These controllers coordinate their actions to collaboratively manage plant-wide real-time operations \cite{christofides2013distributed,stewart2010cooperative,liu2011iterative}.
Compared to conventional centralized paradigms, a distributed structure significantly enhances fault tolerance and improves the organizational and maintenance flexibility of the resulting closed-loop system. 
Distributed designs are also highly scalable, facilitating potential expansions of the system \cite{christofides2013distributed,aiche2019,stewart2010cooperative,li2023aiche,yin2021}. 
Moreover, distributed control can be implemented with multiple computing units, and in this way the overall computational efficiency can be improved compared to that of the centralized counterpart \cite{sutton2018reinforcement, espeholt2018impala}.

Model predictive control (MPC) \cite{kumar2012model,rawlings2017model}, as one of the most widely adopted advanced control methodologies, has been extensively leveraged to develop distributed process control solutions. Specifically, to manage complex industrial processes comprising multiple interacting subsystems, various distributed model predictive control (DMPC) algorithms have been proposed \cite{liu2010distributed}. 
In \cite{magni2006stabilizing}, local controllers can address nonlinear systems with decaying disturbances through contraction constraints and ensure closed-loop stability, provided that interactions among local controllers also exhibit a decaying nature.
Two DMPC designs are presented in \cite{liu2010sequential}, with one design utilizing a sequential communication protocol and the other adopting an iterative communication protocol.
Later, in \cite{liu2011iterative}, the presence of asynchronous and delayed measurement issues in chemical processes was addressed by devising an iterative distributed MPC scheme with the adoption of Lyapunov-based control techniques.
The convergence of DMPC to centralized MPC was considered for switched nonlinear systems in \cite{heidarinejad2013distributed}.
In \cite{zhao2023feature},  reduced-order recurrent neural network models are used in Lyapunov-based DMPC to stabilize the nonlinear system. 
Recently, encrypted DMPC has been used in safety-critical situations with the consideration of networked nature \cite{kadakia2023encrypted,kadakia2024encrypted}.
Meanwhile, it is worth noting that the dependence of MPC/DMPC on precise first-principles models has limited their broader applications. 
If the models are inaccurate or incomplete, their control performance may suffer. 
However, developing accurate models is often challenging and sometimes impossible, especially for complex and nonlinear systems.

In recent years, the development of reinforcement learning (RL) has drawn substantial attention in process control due to its potential to control complex systems with superior performance without relying on first-principles models \cite{lee2018machine, shin2019reinforcement, nian2020review, spielberg2017deep}. 
The underlying systems are typically described using the Markov decision process (MDP), where the next state is dependent only on the current state and action, as has been adopted in \cite{ma2019continuous, wang2024control}.  In the early period, Hoskins et al. used RL to optimize control strategies by learning optimal operations through reward and punishment mechanisms with application to process control \cite{hoskins1992process}.
Most recently, in \cite{wang2024control}, control Lyapunov-barrier functions are combined with RL to guarantee the stability and security of the resulting closed-loop systems.
RL-based distributed control is particularly favorable for large-scale model-free industrial management. 
Recently, some results have been obtained in developing RL-based distributed control for large-scale industrial processes. 
For example, in \cite{foerster2016learning}, a communication protocol among local agents and local control strategies are learned simultaneously, enabling agents to cooperate and achieve global goals. 
However, in these studies, agents are typically homogeneous and often possess global information, which may not always be feasible in the context of industrial control. Furthermore, as a paramount concern in industrial control, the stability of the distributed closed-loop system has been overlooked in the existing literature. To the best of the authors' knowledge, how to ensure closed-loop stability of stochastic nonlinear systems in a RL-based distributed control framework remains an open question.

To address the challenges, we propose a Lyapunov-based RL method for distributed control of a general class of stochastic large-scale nonlinear systems. First, we derive sufficient conditions for the closed-loop system to be stable in a model-dree distributed control setting, which seamlessly enables the analysis of closed-loop stability with only data. Based on this, a practical RL-based distributed control algorithm is proposed. Local controllers and Lyapunov functions are parameterized by neural networks and trained to guarantee the stability of the closed-loop system. Furthermore, the algorithm operates with a distributed training and decentralized execution paradigm, which enables minimal communication among subsystems during the training phase and no communication during the execution phase. Finally, the proposed algorithm is applied to a simulated chemical process and compared with competitive baselines.

The remainder of the paper is organized as follows. Section~\ref{sec:pre} presents the basic definitions, the assumptions, and the problem to be investigated. The sufficient conditions for closed-loop stability in a model-free distributed control setting are derived in Section~\ref{sec:theorem}. In Section~\ref{sec:algorithm}, a distributed actor-critic algorithm is established based on these sufficient Lyapunov-based conditions. Following the algorithm, simulations are made on the chemical process and the results are shown in Section~\ref{sec:experiment}. The conclusion and potential future directions are given in Section~\ref{sec:conclusion}.

\section{Preliminaries and problem formulation}\label{sec:pre}

\subsection{Notation}

$\mathbb{R}$ denotes the set of real numbers. $\mathbb{R}^+$ denotes the set of non-negative real numbers. 
$\mathbb{R}^n$ denotes the $n$ dimensional Euclidean space.
$\mathbb{Z}^+$ denotes the set of non-negative integers. 
$\times$ denotes the Cartesian product.
$\odot$ denotes the Hadamard product.
$\mathcal{N}(\mu, \sigma^2)$ represents a Gaussian distribution with mean $\mu\in \mathbb{R}^n$ and standard deviation $\sigma \in\mathbb{R}^n$.
$\text{U}(a,b)$ represents a uniform distribution with $a$ and $b$ as the lower and upper bounds, respectively.

In this paper, we consider the control of a general class of stochastic nonlinear systems characterized by the Markov decision process (MDP) $\s{M}$, wherein the next state of the system $s_{k+1}\in\s{S}$ depends solely on the current state $s_{k}\in\s{S}$ and action $a_{k}\in\s{A}$, where $\s{S} \subseteq \m{R}^n$ denotes the state space; $\s{A} \subseteq \m{R}^m$ denotes the action space; $k\in \m{Z}^+$ denotes the time instant. The transition of state is dominated by the conditional probability density function (PDF) $P(s_{k+1}|s_k,a_k)$, which denotes the conditional probability density of the next state $s_{k+1}$ given the current state $s_k$ and action $a_k$. 
The stage cost $c_k = \mathcal{C}(s_{k})$ measures the instantaneous performance of the state $s_k$.
$\rho(s_0)$ is the PDF that characterizes the distribution of the initial state $s_0$.
The stochastic control policy $\pi(a_k|s_k)$ of MDP is a conditional PDF of the action $a_k$ given the current state $s_k$. 

Given a control policy $\pi$, $P_\pi(s_{k+1}|s_k)\triangleq \int_{a_k\in\s{A}} \pi(a_k|s_k)P(s_{k+1}|a_k,s_k) \dif{a_k}$ denotes the closed-loop state transition PDF.
The conditional probability density $P(s_{k} | \rho, \pi, k)$ of the state $s_k$ given initial distribution $\rho$, control policy $\pi$ and time instant $k$ is defined recursively as follows: 
\begin{subequations}
\begin{align}
& P(s_{0}|\rho, \pi, 0)=\rho(s_{0})  \\[0.3em]
& P(s_{k+1} | \rho, \pi, k+1) = \int_{s_k\in\s{S}} P_\pi(s_{k+1}| s_{k}) P(s_k|\rho, \pi, k) \dif{s_{k}}
\end{align}
\end{subequations}
where $k\in\m{Z}^+$.

\subsection{Subsystem description for the MDP}

Assume the MDP $\s{M}$ can be partitioned into $\nu\in \mathbb{Z}^{+}$ interconnected subsystems. We define $\s{V}:=\big\{1,\ldots,\nu\big\}$. 
For any $i\in\s{V}$, $s^i\in\s{S}^i$ and $a^i \in \s{A}^i$ denote the state and action of the $i$th subsystem, respectively; 
the state space of $\s{M}$ is composed of the state spaces of the subsystems, that is, $\s{S}=\s{S}^1\times \cdots \times \s{S}^\nu$ and $s=[(s^1)^{\top}, \cdots, (s^\nu)^{\top}]^{\top}$;
its action space is also composed of the action spaces of the subsystem, that is, $\s{A}=\s{A}^1\times \cdots \times \s{A}^\nu$ and $a=[(a^1)^{\top}, \cdots, (a^\nu)^{\top}]^{\top}$;
for any $i\in\s{V}$, $s^{-i}_k$ and $a^{-i}_k$ denote the states and actions of subsystems other than those of the $i$th subsystem.
$P^i(s_{k+1}^i|s^i_k, a^i_k)$ is the state transition PDF for the $i$th subsystem, which is not Markov since its value is influenced by $s^{-i}_k$ and $a^{-i}_k$ as \cite{canese2021multi} mentions.
$c_k^i=\s{C}^i(s_k^i)$ is the stage cost for the $i$th subsystem and satisfies $\s{C}(s) = \sum_{i=1}^{\nu}\s{C}^i(s^i)$.
$\rho^i$ specifies the initial state PDF for the $i$th subsystem, satisfying $\rho(s)=\prod_{i=1}^\nu \rho^i(s^i)$.

In this paper, we focus on a full-state tracking problem and denote the reference state as $s_{ref}$, where the stage cost for the $i$th subsystem, $i\in\s{V}$, is designed as the square norm of local state tracking error:
\begin{equation}
\label{eq:cost}
    \s{C}^i(s^i_k) \triangleq \Vert s^i_{k}-s^i_{ref} \Vert^2
\end{equation}
Accordingly, the stage cost of full-state is given by $\s{C}(s_k)=\sum_{i=1}^\nu \s{C}^i(s_k^i)=\Vert s_{k}-s_{ref} \Vert^2$.

\subsection{Problem formulation}

In this paper, we aim to control the stochastic nonlinear system described by $\s{M}$ using a RL-based distributed method with guaranteed closed-loop stability.  For the $i$th subsystem, $i\in \s{V}$, a control policy $\pi^i(a_k^i|s_k^i)$ is learned, and the resulting distributed control policy $\pi_d$ for the entire system is:
\begin{equation}
    \pi_d(a_k|s_k) \triangleq \prod_{i=1}^\nu \pi^i(a^i_k|s^i_k)
\end{equation} 
We use $P^i(s_k^i|\rho^i, \pi^i,k)$ to denote the conditional distribution of state $s_k^i$ with the given initial state distribution $\rho^i$, control policy $\pi^i$ and time instant $k$ for the $i$th subsystem.

The stability criterion, defined following the concept of mean square stability \cite{shaikhet1997necessary, bolzern2010markov, han2020actor}, is presented as follows:
\begin{mydef}
    Consider the MDP $\s{M}$ comprising $\nu$ subsystems. If there exists a distributed control policy $\pi_d$ such that the resulting closed-loop system is ergodic, and the system's stage cost converges to zero, as shown by:
    \[
    \lim_{k\to \infty} \mathbb{E}_{s_k\sim P(\cdot|\rho, \pi_d, k}) \s{C}(s_k)=0
    \]
    then the MDP $\s{M}$ is said to be stable in mean cost under the distributed control policy $\pi_d$.
\end{mydef}

The control problem to be addressed in this paper is to learn a set of local control policies such that the MDP $\mathcal{M}$ is stable in mean cost with the model-free distributed control policy $\pi_d$.

\section{Stability analysis for distributed control}
\label{sec:theorem}

In this section, we present the key theorem for validating the closed-loop stability of the MDP $\s{M}$ controlled by the distributed control policy $\pi_d$. First, we assume that there exists a distributed control policy such that the Markov chain induced is ergodic, as commonly assumed in the existing RL literature \cite{han2020actor, sutton2008convergent, korda2015td, bhandari2018finite}.
\begin{myass}\label{assumption2}
There exists a distributed control policy $\pi_d$ such that the MDP $\s{M}$ controlled by $\pi_d$ is ergodic and exhibits a unique stationary distribution $q_{\pi_d}(s)$, which is defined as:
    \begin{equation}
    \label{eq:ergodic}
        q_{\pi_d}(s) \triangleq \lim_{k\to\infty} P(s|\rho, \pi_d, k)
    \end{equation}
\end{myass}

{
\begin{myremark}
    From a theoretical point of view, an MDP controlled by a policy is ergodic if the closed-loop system is irreducible (i.e., every state can be reached from any other state) and aperiodic (i.e., the system does not get trapped in cycles) \cite{puterman2014markov}. 
    Nonetheless, verifying the ergodicity of an MDP with a control policy is challenging in practice. This property may be approximately assessed by simulating the system over a sufficiently long period and examining the convergence of the state distribution. Statistical tests, such as the Gelman-Rubin diagnostic \cite{TH2008},  can be applied to assess convergence characteristics.
\end{myremark}}

This unique distribution $q_{\pi_d}(s)$ captures the long-term behavior of the ergodic system subject to the distributed control policy $\pi_d$.  
In what follows, we leverage the Lyapunov theorem to establish sufficient conditions for the MDP $\s{M}$ controlled by $\pi_d$ to be stable in mean cost. In well-established control theory, the asymptotic stability of a deterministic system can be guaranteed if an energy-like Lyapunov function of the system keeps decreasing along state trajectories \cite{khalil2009lyapunov, raghunathan2013optimal}. However, such an energy-decreasing condition is unsuitable for stability analysis of stochastic nonlinear systems without explicit models. In \cite{han2020actor}, a stochastic stability theorem for model-free control method is developed based on data.
Nonetheless, previous results are not applicable to distributed control methods.
This paper aims to identify a set of Lyapunov functions $L^i: S^i\to \m{R}^{+}$ for each subsystem and analyze the closed-loop stability in a stochastic manner. 
The sufficient conditions for the MDP $\s{M}$ composed of $\nu$ subsystems to be stable in mean cost are derived as follows:

\begin{myth}
\label{th:lya}
Let Assumption  \ref{assumption2} hold.
The MDP $\s{M}$ controlled by the distributed control policy $\pi_d$ is stable in mean cost
if there exist positive constants $\alpha_{1}$, $\alpha_{2}$, and $\alpha_{3}$, and 
a set of Lyapunov functions $L^i:\s{S}^i \to\mathbb{R}^+$ for each $i\in\s{V}$,  
such that:
\begin{subequations}\label{eq:lyappunov}
\begin{align}
&\alpha_{1}\s{C}^i(s_k^i)\leq L^i(s_k^i)\leq \alpha_{2}\s{C}^i(s_k^i) \label{eq:lyapunov1}\\[0.3em]
&L(s_k)=\sum_{i=1}^\nu L^i(s_k^i) \label{eq:lyapunov2}\\[0.3em]
&\mathbb{E}_{s_k\sim u_{\pi_d}(\cdot),  s_{k+1}\sim P_{\pi_d}(\cdot|s_k)}[L(s_{k+1})-L(s_k)]\leq -\alpha_3\mathbb{E}_{s_k\sim u_{\pi_d}(\cdot)} \s{C}(s_k)\label{eq:lyapunov3}
\end{align}
\end{subequations}
where $u_{\pi_d}(s) \triangleq \lim_{N\to\infty}\frac{1}{N}\sum^{N}_{t=0}P(s|\rho, \pi_d, t)$ is the average probability density of the state $s$.
\end{myth}

\begin{proof}

If Assumption~\ref{assumption2} is satisfied, then the MDP $\mathcal{M}$ has a unique stationary state probability density $q_{\pi_d}(s)$ as described in (\ref{eq:ergodic}).
According to the Abelian theorem, the state distribution $u_{\pi_d}(s)$ will converge to state distribution $q_{\pi_d}(s)$ as $N\rightarrow \infty$, that is, for any state $s\in \s{S}$, we have $u_{\pi_d}(s)=q_{\pi_d}(s)$. Consequently, (\ref{eq:lyapunov3}) is equivalent to:
$$\mathbb{E}_{s_k\sim u_{\pi_d}(\cdot), s_{k+1}\sim P_{\pi_d}(\cdot|s_k)}[L(s_{k+1})-L(s_k)]\leq -\alpha_3\mathbb{E}_{s_k\sim q_{\pi_d}(\cdot)} \s{C}(s_k)$$
Further, the inequation can be re-written as:
{\small
\begin{equation}\label{eq:lya}
    \int_{s_k\in\s{S}} \lim_{N\to\infty}\frac{1}{N} \sum_{t=0}^N P(s_k|\rho, \pi_d, t) \mathbb{E}_{s_{k+1}\sim P_{\pi_d}(\cdot|s_k)}[L(s_{k+1})-L(s_k)] \dif{s_{k}} \leq -\alpha_3 \m{E}_{s_k\sim q_{\pi_d}(\cdot)} \s{C}(s_k)
\end{equation}}
According to~(\ref{eq:lyapunov1}), the Lyapunov function $L^i$ is bounded by $\alpha_2\s{C}^i$. 
One has:
    $$P(s_k|\rho, \pi_d, t)L(s_k) \leq \alpha_2 P(s_k|\rho, \pi_d, t) \sum_{i=1}^{\nu}\s{C}^i(s_k^i)$$
where the right-hand side is an integrable function. Meanwhile, $\{\frac{1}{N}\sum_{t=0}^N P(s_k|\rho, \pi_d, t)L(s_k)\}$ converges point-wise to $q_{\pi_d}(s_k)L(s_k)$ as $N$ approaches $\infty$. Therefore, according to the Lebesgue Dominated Convergence Theorem \cite{blackwell1963converse}, the left-hand-side of (\ref{eq:lya}) can be re-written as:
    $$\begin{aligned} 
     & \int_{s_k\in\s{S}} \lim_{N\to\infty}\frac{1}{N} \sum_{t=0}^N P(s_k|\rho, \pi_d, t) \mathbb{E}_{s_{k+1}\sim P_{\pi_d}(\cdot|s_k)}[L(s_{k+1})-L(s_k)] \dif{s_{k}}\\
      =&\lim_{N\to\infty}\frac{1}{N}\int_{s_k\in\s{S}}\mathbb{E}_{s_{k+1}\sim P_{\pi_d}(\cdot|s_k)}\sum^{N}_{t=0}[L(s_{k+1})-L(s_{k})]P(s_k|\rho, \pi_d, t)\dif{s_k}\\
     =& \lim_{N\to\infty}\frac{1}{N} (\sum^{N+1}_{t=1} \mathbb{E}_{s_k\sim P(\cdot|\rho, \pi_d, t)} L(s_k)-\sum^{N}_{t=0} \mathbb{E}_{s_k\sim P(\cdot|\rho, \pi_d, t)} L(s_k))\\
      =& \lim_{N\to\infty}\frac{1}{N} \left(\mathbb{E}_{s_k\sim P(\cdot|\rho, \pi_d, N+1)} L(s_{k})-\mathbb{E}_{s_k\sim \rho(\cdot)} L(s_k)\right)
\end{aligned}$$
Consequently, it can be obtained that:
\begin{equation}\label{equation:lim}
    \lim_{N\to\infty}\frac{1}{N} \left(\mathbb{E}_{s_k\sim  P(\cdot|\rho, \pi_d, N+1)} L(s_k)-\mathbb{E}_{s_k\sim \rho(\cdot)} L(s_k)\right) \leq -\alpha_3\mathbb{E}_{s_k\sim q_{\pi_d}(\cdot)}\s{C}(s_k)
\end{equation}
Since $L(s_k)\geq 0$, (\ref{equation:lim}) leads to:
$$\mathbb{E}_{s_k\sim q_{\pi_d}(\cdot)}\s{C}(s_k)\leq \frac{1}{\alpha_3}\lim_{N\to\infty}\frac{1}{N} \mathbb{E}_{s_k\sim \rho(\cdot)} L(s_k)$$
Furthermore, since $\mathbb{E}_{s_k\sim \rho(\cdot)} L(s_k)$ is bounded by $\alpha_2 \s{C}(s_k)$, it holds that:
   $$ \mathbb{E}_{s_k\sim q_{\pi_d}(\cdot)}\s{C}(s_k)=0$$
which implies $\lim_{k\to \infty} \mathbb{E}_{s_k\sim P(\cdot|\rho, \pi_d, k}) \s{C}(s_k)=0$ and concludes the proof.
\end{proof}

Theorem~\ref{th:lya} provides sufficient conditions for the MDP $\mathcal M$ comprising multiple interconnected subsystems to be stable in mean cost. Conditions \eqref{eq:lyapunov2} and~\eqref{eq:lyapunov3} essentially require that the expected value of the Lyapunov function for the whole stochastic system consistently decreases. This requirement can be verified through extensive data analysis, allowing for its integration into RL frameworks. Furthermore, the validation of \eqref{eq:lyapunov3} requires only a minimal level of information exchange among the subsystems by utilizing a set of Lyapunov functions for subsystems, which will be shown in Section~\ref{sec:algorithm}.

\section{Distributed Lyapunov actor-critic method}
\label{sec:algorithm}
\begin{figure}
    \centering
    \includegraphics[width=0.8\textwidth]{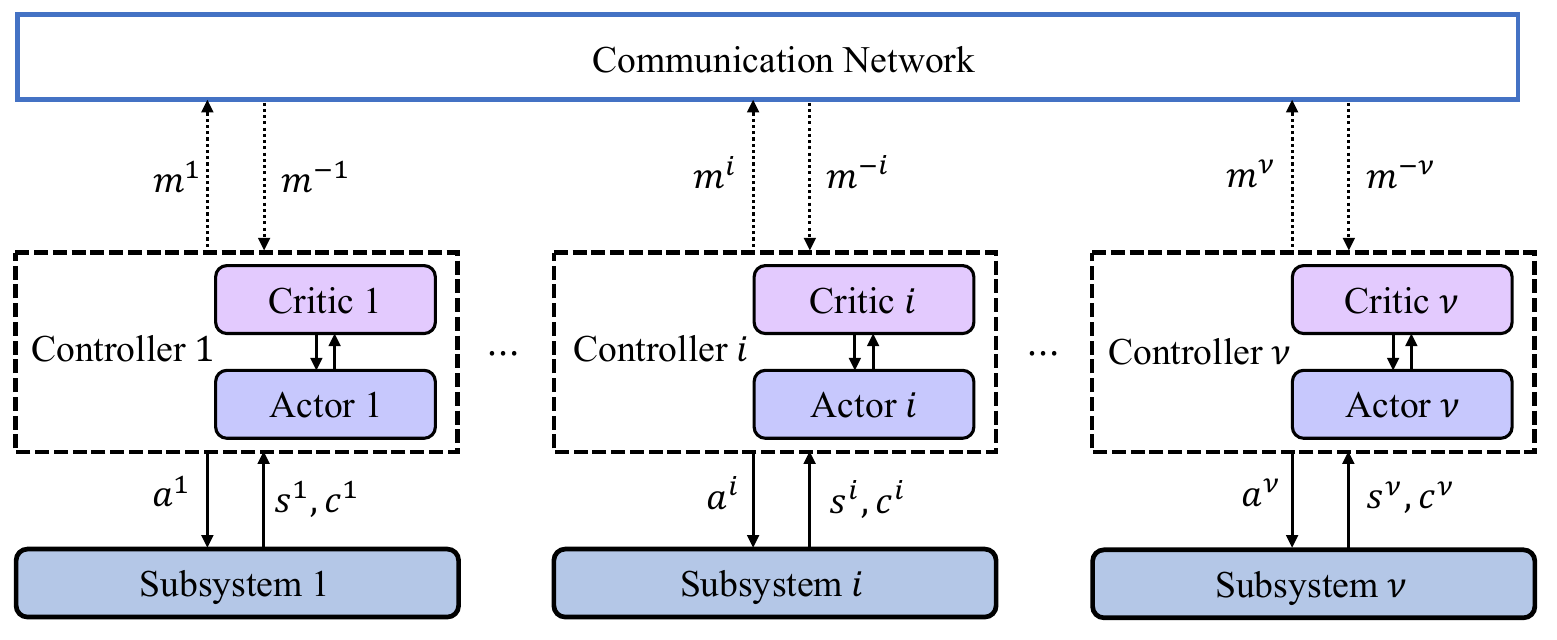}
    \caption{A graphic illustration of the distributed training paradigm.}
    \label{fig:scheme}
\end{figure}

We utilize Theorem~\ref{th:lya} to guide the design of a distributed Lyapunov actor-critic (DLAC) scheme, with a distributed training and decentralized execution paradigm. As illustrated in Figure \ref{fig:scheme}, during the distributed training phase, an RL-based controller is deployed for each subsystem. Controller $i$ for the $i$th subsystem, $i\in\mathcal V$, consists of a critic and an actor, denoted by Critic $i$ and Actor $i$, respectively. 
The critics are trained to evaluate the future cumulative discounted stage costs for the subsystems given local current states and actions. 
Under the guidance of the critics, the actors for the subsystems are trained to minimize this expected cumulative discounted stage costs.

During the training phase, each controller interacts with its corresponding subsystem through actions generated by its actor and state feedback from the subsystem. Meanwhile, local controllers exchange information with each other through the communication network. Based on the sampled local data and the exchanged information, actors and critics update their policies online.
Upon convergence of the training phase for the critics and actors, further information exchange among the local controllers becomes unnecessary, allowing each controller to operate in a decentralized fashion. Ultimately, these RL controllers, though trained in a distributed manner, support decentralized execution during online implementation.

\subsection{Distributed training of the critics}
The critics of the subsystems play two important roles. First, they evaluate the future cumulative discounted stage costs for the subsystems given local state-action pairs and guide the training of the actors to optimize these local values. Second, the evaluation functions of critics provide the basis for the actors to ensure the Lyapunov-based conditions in Theorem~\ref{th:lya} and facilitate the analysis of the closed-loop stability of the whole system. 

Similar to the centralized actor-critic framework \cite{sutton2018reinforcement}, for the $i$th subsystem, $i\in\s{V}$, the critic evaluation function 
$Q_{\pi^i}^i$ takes the local state $s^i_k$ and action $a^i_k$ as inputs and estimates the expected future discounted accumulated costs following policy $\pi^i$, which is defined as:
\begin{equation}\label{eq:bellman}
    Q^i_{\pi^i}(s^i_k, a^i_k) \triangleq \s{C}^i(s^i_k)+ \gamma \m{E}_{a^i_{k+1}\sim \pi^i(\cdot|s^i_{k+1}), s_{k+1}^i\sim P^i(\cdot|s_k^i, a_k^i)}\left[Q^i_{\pi^i}(s^i_{k+1}, a^i_{k+1})\right]
\end{equation}
where $0<\gamma<1$ is a discount factor.  
Note that due to the interactions among subsystems, the value of $Q^i_{\pi^i}$ for the $i$th subsystem, $i\in\mathcal V$, is influenced by the controllers of its interacting subsystems through $P^i$. 

Considering the continuity of state and action spaces, we employ a neural network $\hat{Q}^i(s_k^i, a_k^i; \tilde{w}^i)$, where $\tilde{w}^i$ denotes the trainable parameters, to approximate the critic evaluation function $Q_{\pi^i}^i$. The neural networks are trained to minimize the temporal difference error derived from \eqref{eq:bellman}. Formally, the loss function for the $i$th critic neural network, $i\in\s{V}$, is designed as follows:
\begin{equation}
    \texttt{J}^i_{critic}(\tilde{w}^i) = \frac{1}{2} \m{E}_{(s^i_k, a^i_k, c^i_k, s^i_{k+1})\sim \s{D}^i, a_{k+1}^i\sim \pi^i(\cdot|s_{k+1}^i)} [(\hat{Q}^i(s_k^i, a_k^i; \tilde{w}^i) - c^i_{k} - \gamma \hat{Q}^i(s^i_{k+1}, a^i_{k+1}; \tilde{w}^i))^2]
    \label{loss_c}
\end{equation}
where $\s{D}^i\triangleq\{(s^i_k, a^i_k, c^i_k, s^i_{k+1})\}_{k=1}^N$ denotes the dataset composed of data sampled from the $i$th subsystem, with $N$ being the size of dataset.

To incorporate conditions provided in Theorem~\ref{th:lya} into the distributed actor-critic framework, we designate the state value as the local Lyapunov function for each subsystem. The state value function effectively estimates the future cumulative discounted costs and can seamlessly integrate Lyapunov approaches into the RL framework. Previous studies have demonstrated the viability of the state value function as a Lyapunov function \cite{han2020actor,chow2018lyapunov}. Thus, we formalize the local Lyapunov function for the $i$th subsystem, $i\in\s{V}$, at state $s^i_k$ as:
\begin{equation}
    L^i(s_k^i) = \m{E}_{a^i_k\sim \pi^i(\cdot|s_k^i)}\hat{Q}^i(s_k^i, a_k^i; \tilde{w}^i)\label{eq:eq}
\end{equation}

According to the Robbins-Monro Theorem \cite{robbins1951stochastic}, the local Lyapunov function $L^i(s_k^i)$, defined as the expectation over actions sampled from the policy $\pi^i(a^i_k|s_k^i)$, can be approximated in practice by sampling a single action $a_k^i$ from policy $\pi^i(a^i_k|s_k^i)$. Thus, $L^i(s_k^i)$ can be computed effectively as $\hat{Q}^i(s_k^i, a_k^i; \tilde{w}^i)$, using the sampled action $a_k^i$,
which obviates the need for explicit expectation calculations during implementation.

\subsection{Distributed training of the actors}\label{actor}

We approximate the actor policy for the $i$th subsystem, $i\in\s{V}$, by a neural network parameterized by $\tilde{\theta}^i$.
This neural network models the policy $\pi^i$ as a Gaussian distribution, encoding the state $s^i_k$ into the action mean $u_k^i(s^i_k; \tilde{\theta}^i)\in \mathbb{R}^n$ and the action standard deviation $\sigma_k^i(s^i_k; \tilde{\theta}^i)\in \mathbb{R}^n$.
To generate action instances, the model introduces a noise vector $\epsilon^i_k\in \mathbb{R}^n$, sampled from a multivariate normal distribution $\mathcal{N}(\mathbf{0}, \mathbf{I})$, resulting in the action instance $a^i_k = f^i(s^i_k; \epsilon^i_k, \tilde{\theta}^i) = u_k^i(s^i_k; \tilde{\theta}^i) + \epsilon^i_k\odot \sigma_k^i(s^i_k; \tilde{\theta}^i)$.
Consequently, the action $a^i_k$ for the $i$th subsystem follows the encoded Gaussian distribution $\mathcal{N}(u_k^i(s^i_k; \tilde{\theta}^i), (\sigma_k^i(s^i_k; \tilde{\theta}^i))^2)$.

The Lyapunov-based conditions proposed in Theorem~\ref{th:lya} and the critic evaluation functions are used to establish the objective functions of the actors. The overall objective of actors is to find a set of parameters $\{\tilde{\theta}^i\}_{i=1:\nu}$ such that condition \eqref{eq:lyapunov3} is ensured, that is,
\begin{equation}
\label{eq:condition1}
    \sum_{i=1}^{\nu}\m{E}_{(s^i_k, a^i_k, c^i_k, s^i_{k+1})\sim\s{D}^i}\left[\hat{Q}^i(s^i_{k+1}, f^i(s^i_{k+1}; \epsilon^i_{k+1},\tilde{\theta}^i);\tilde{w}^i)-\hat{Q}^i(s^i_k, a^i_k;\tilde{w}^i)+\alpha_3  {c^i_{k}}\right]\leq 0
\end{equation}

For any $i\in\s{V}$, the training of Actor $i$ and Critic $i$ is conducted in a distributed manner.
In this setup, the information that the controller of the $i$th subsystem sends to the communication network is designed to be a scalar variable evaluating the energy-decreasing value within the subsystem, that is,
\begin{equation}
    m^i_k \triangleq \hat{Q}^i(s^i_{k+1}, f^i(s^i_{k+1}; \epsilon^i_{k+1}, \tilde{\theta}^i);\tilde{w}^i)-\hat{Q}^i(s^i_k, a^i_k;\tilde{w}^i)+\alpha_3  {c^i_{k}}
    \label{eq:message}
\end{equation} 
The information that the controller of the $i$th subsystem receives from the communication network is the sum of the scalar variables shared by the other controllers, that is,
\begin{equation}
    m^{-i}_k = \sum_{j\in\s{V},j\neq i} m^j_k
\end{equation}
Following this, the constraint (\ref{eq:condition1}) for the $i$th subsystem can be re-written as an inequality with information known to Actor $i$ as:
\begin{equation}
    \m{E}_{(s^i_k, a^i_k, c^i_k, s^i_{k+1})\sim\s{D}^i}\left[m^{-i}_k+ \hat{Q}^i(s^i_{k+1}, f^i(s^i_{k+1}; \epsilon^i_{k+1},\tilde{\theta}^i);\tilde{w}^i)-\hat{Q}^i(s^i_k, a^i_k;\tilde{w}^i)+\alpha_3  {c^i_{k}}\right] \leq 0
\end{equation}

To facilitate exploration during training, the entropy of the local control policy is constrained to be greater than or equal to a threshold $\s{E}^i$, that is,
\[\m{E}_{(s^i_k, a^i_k)\sim\s{D}^i} [-\log(\pi^i( f^i(s^i_k;\epsilon^i_{k},\tilde{\theta}^i)|s^i_k))] \geq \s{E}^i\]
where $\s{E}^i$ is a manually set threshold. This constraint ensures that the encoded action distribution for the $i$th subsystem maintains a certain level of stochastic property. $\s{E}^i$ is usually set as the negative of the number of local control variables \cite{han2020actor, haarnoja2018soft}.

To sum up, for any $i\in\s{V}$, the actor optimization problem for the $i$th subsystem can be described as follows:
\begin{subequations}\label{eq:actor}
\begin{align}
\text{find }&\tilde{\theta}^i  \\
\text{s.t. }    & \m{E}_{(s^i_k, a^i_k, c^i_k, s^i_{k+1})\sim\s{D}^i}\left[m^{-i}_k+ \hat{Q}^i(s^i_{k+1}, f^i(s^i_{k+1}; \epsilon^i_{k+1}, \tilde{\theta}^i);\tilde{w}^i)-\hat{Q}^i(s^i_k, a^i_k;\tilde{w}^i)+\alpha_3  {c^i_{k}}\right] \leq 0  \\
               &  \m{E}_{(s^i_k, a^i_k)\sim \s{D}^i} [-\log(\pi^i(f^i(s^i_k;\epsilon_k^i, \tilde{\theta}^i)|s^i_k))] \geq \s{E}^i \label{eq:condition2}
\end{align}
\end{subequations}

To solve the above constrained problem, we use the Lagrange method to transform \eqref{eq:actor} into an unconstrained optimization problem, where the loss function for the Actor $i$ is:
\begin{equation}
\label{eq:loss_actor}
    \texttt{J}^i_{actor}(\tilde{\theta}^i) = \m{E}_{(s^i_k, a^i_k, c^i_k, s^i_{k+1})\sim\s{D}^i} \left[
         e^{\beta^i}\log(\pi^i( f^i(s^i_k;\epsilon_k^i,\tilde{\theta}^i)|s^i_k))+  
         e^{\lambda^i}\hat{Q}^i({s^i_{k+1}}, f^i({s^i_{k+1}};\epsilon_{k+1}^i,\tilde{\theta}^i);\tilde{w}^i)
  \right]
\end{equation}
where $\beta^i, \lambda^i\in \mathbb{R}$; $e^{\beta^i}$ and $e^{\lambda^i}$ serve as Lagrange multipliers. The updates of these multipliers are controlled by minimizing their respective loss functions described below:
\begin{subequations}
\begin{align}
        &\texttt{J}^i_{actor}(\beta^i) =   -\beta^i\m{E}_{(s^i_k, a^i_k)\sim \s{D}^i} \left[\log(\pi^i(f^i(s_k^i; \epsilon_k^i, \tilde{\theta}^i)|s^i_k))+ \s{E}^i\right] \label{eq:loss_beta}\\
        &\texttt{J}^i_{actor}(\lambda^i) =   -\lambda^i\m{E}_{(s^i_k, a^i_k, c^i_k, s^i_{k+1})\sim\s{D}^i} \left[m^{-i}_k+\hat{Q}^i(s^i_{k+1}, f^i(s^i_{k+1};\epsilon_{k+1}^i,\tilde{\theta}^i); \tilde{w}^i)-\hat{Q}^i(s^i_k, a^i_k; \tilde{w}^i)+\alpha_3 c^i_{k}\right]\label{eq:loss_lambda}
    \end{align}    
\end{subequations}
The actor neural network parameter $\tilde{\theta}^i$ and Lagrange multipliers $\beta^i, \lambda^i$ are all updated using the stochastic gradient descent method with their respective loss functions. Based on the derivations above, the pseudo-code for the proposed DLAC algorithm is presented in Algorithm~\ref{algorithm}.

\begin{algorithm}
    \label{algorithm}
    \caption{The training algorithm for distributed Lyapunov actor-critic}
    \KwData{number of subsystems $\nu$ and other hyperparameters}
    \KwResult{A set of control policies $\{\pi^i$, $\forall i\in\s{V}\}$}
    Initialize parameters $\tilde{w}^i$, $\tilde{\theta}^i$, learning rates, dataset $\mathcal{D}^i=\emptyset$, $\forall i\in\s{V}$, and $j=0$\\
    \While{$j< n_{eps\_max}$}{
    Sample an episode with $n_{steps}$ length using the set of all controllers\\
    $j=j+1$\\
    \For{each $\mathcal{D}^i$}
    {Insert samples $\{(s^i_k, a^i_k, c^i_{k}, s^i_{k+1})\}_{k=0:n_{step}-1}$ into $\mathcal{D}^i$}
    \If{ $j \geq n_{eps\_min}$}{
    \If{$j$ is divisible by $n_{eps\_interval}$}
    {
         \Repeat{
            Each controller uses data randomly sampled from $\s{D}^i$ with batch size $n_{batch}$\\
            \For{each Critic $i$}
            {
                Update $\tilde{w}^i$ with the gradient of $\texttt{J}^i_{critic}(\tilde{w}^i)$\\
                Calculate $m_k^i$ by (\ref{eq:message})\\
            }
            \For{each Actor $i$}
            {
                $m^{-i}_k=\sum_{j\in\s{V}, j\neq i} m_k^j$\\
                Update $\beta^i$ with the gradient of $\texttt{J}^i_{actor}(\beta^i)$\\
                Update $\lambda^i$ with the gradient of $\texttt{J}^i_{actor}(\lambda^i)$\\
                Update $\tilde{\theta}^i$ with the gradient of $\texttt{J}^i_{actor}(\tilde{\theta}^i)$
            }            
         }
         Evaluate the performance of the controllers using test data
    }}}
\end{algorithm}

{
\begin{myremark}
The proposed method employs a distributed training and decentralized execution paradigm. During the training phase, the parameters of the neural networks for the local controllers are updated iteratively. During the online implementation phase, the local controllers operate independently, without any real-time information exchange. 
The output measurements from the subsystems are synchronously sampled. The local controllers are executed simultaneously at each sampling instant during implementation.
\end{myremark}}

{\begin{myremark}
    The current work aims to develop an RL-based distributed control scheme, of which the feasible solution can guarantee the stability of the closed-loop system. 
    Specifically, Theorem~\ref{th:lya} presents sufficient conditions for closed-loop stability in mean cost, which are used to guide the training of local controllers in the scheme.
    Meanwhile, the proposed method does not guarantee the optimality of the distributed control strategy, in the absence of information exchange among local controllers during online implementation.
\end{myremark}}

\section{Simulation results}
In this section, we evaluate the performance of the proposed method on a simulated chemical process. 

\label{sec:experiment}
\subsection{Process description}
\begin{figure}[tb]
    \centering
    \includegraphics[width=0.75\textwidth]{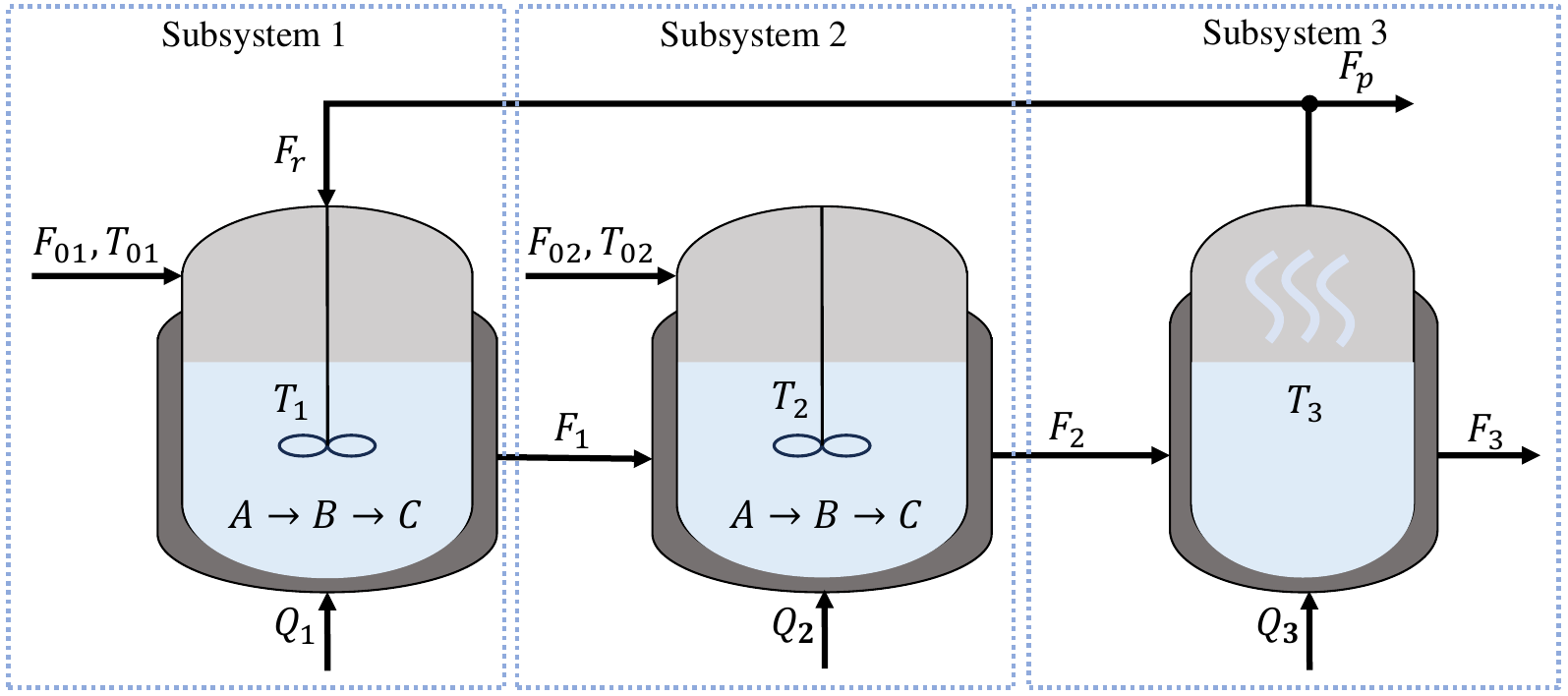}
    \caption{A schematic representation of the chemical process containing two reactors and a separator.}
    \label{fig:three tank}
\end{figure}
We use a benchmark chemical process to evaluate the proposed method. This process 
comprises two continuous stirred-tank reactors and a flash tank separator \cite{liu2009distributed, zhang2013distributed}, with a schematic shown in Figure~\ref{fig:three tank}. Two reactions take place in parallel: the conversion of reactant $A$ to the desired product $B$ ($A \to B$), and the conversion of the desired product to a side product $C$ ($B \to C$). 

In this process, a stream containing reactant $A$ is fed into the first reactor at flow rate $F_{01}$ and temperature $T_{01}$. The outlet of the first reactor enters the second reactor at flow rate $F_{1}$ and temperature $T_1$. The second reactor also receives an additional feed of $A$ at flow rate $F_{02}$ and temperature $T_{02}$. 
The output from the second tank is sent to a separator at flow rate $F_2$ and temperature $T_2$. This separator processes the effluent and recycles a portion of it back to the first tank at flow rate $F_r$ and temperature $T_3$. Each vessel is equipped with a heating jacket that can either supply or remove heat, with the heating input rate for the $i$th vessel denoted by $Q_i, i=1,2,3$. 
The state variables of this nonlinear process include the mass fractions of $A$ and $B$ (denoted by $x_{Ai}$ and $x_{Bi}$), and the temperatures (denoted by $T_i$, $i=1,2,3$) in the three tanks. That it, $s = [x_{A1}, x_{B1}, T_1, x_{A2}, x_{B2}, T_2, x_{A3}, x_{B3}, T_3]^\top$. Nine ordinary differential equations, established based on material and energy balances, describe the dynamic behaviors of this process \cite{liu2009distributed, zhang2013distributed}:
\begin{subequations}\label{eq:three tank}
\begin{align}
& \displaystyle{\frac{dx_{A1}}{dt}} = \frac{F_{10}}{V_1}(x_{A10} - x_{A1}) + \frac{F_r}{V_1}(x_{Ar} - x_{A1}) - k_1 e^{\frac{-E_1}{rT_1 }}x_{A1}  \\[0.3em]
& \displaystyle{\frac{dx_{B1}}{dt}} = \frac{F_{10}}{V_1}(x_{B10} - x_{B1}) +  \frac{F_r}{V_1}(x_{Br} - x_{B1})+  k_1 e^{\frac{-E_1}{rT_1 }}x_{A1} - k_2 e^{\frac{-E_2}{rT_1 }}x_{B1} \\[0.3em]
& \displaystyle{~~\frac{dT_{1}}{dt}} = \frac{F_{10}}{V_1}(T_{10} - T_{1}) + \frac{F_r}{V_1}(T_{3} - T_{1}) - \frac{\Delta H_1}{ c_p}k_{1}e^{\frac{-E_1}{rT_1}}x_{A1} - \frac{\Delta H_2}{ c_p}k_{2}e^{\frac{-E_2}{rT_1}}x_{B1} +  \frac{Q_1}{\rho c_pV_1} \\[0.3em]
& \displaystyle{\frac{dx_{A2}}{dt}} = \frac{F_{1}}{V_2}(x_{A1} - x_{A2}) + \frac{F_{20}}{V_2}(x_{A20} - x_{A2}) - k_1 e^{\frac{-E_1}{rT_2 }}x_{A2}  \\[0.3em]
& \displaystyle{\frac{dx_{B2}}{dt}} = \frac{F_{1}}{V_2}(x_{B1} - x_{B2}) + \frac{F_{20}}{V_2}(x_{B20} - x_{B2})+  k_1 e^{\frac{-E_1}{rT_2 }}x_{A2} - k_2 e^{\frac{-E_2}{rT_2}}x_{B2}, \\[0.3em]
& \displaystyle{~~\frac{dT_{2}}{dt}} = \frac{F_{1}}{V_2}(T_{1} - T_{2}) + \frac{F_{20}}{V_2}(T_{20} - T_{2}) - \frac{\Delta H_1}{ c_p}k_{1}e^{\frac{-E_1}{rT_2}}x_{A2} - \frac{\Delta H_2}{ c_p}k_{2}e^{\frac{-E_2}{rT_2}}x_{B2} +  \frac{Q_2}{\rho c_pV_2} \\[0.3em]
& \displaystyle{\frac{dx_{A3}}{dt}} = \frac{F_{2}}{V_3}(x_{A2} - x_{A3}) - \frac{(F_r+F_p)}{V_3}(x_{Ar} - x_{A3}) \\[0.3em]
& \displaystyle{\frac{dx_{B3}}{dt}} = \frac{F_{2}}{V_3}(x_{B2} - x_{B3}) - \frac{(F_r+F_p)}{V_3}(x_{Br} - x_{B3}) \\[0.3em]
& \displaystyle{~~\frac{dT_{3}}{dt}} = \frac{F_{2}}{V_3}(T_{2} - T_{3}) + \frac{Q_3}{\rho c_pV_3}+\frac{(F_{r}+F_{p})}{\rho c_{p}V_{3}}(x_{Ar}\Delta H_{\text{vap1}}+x_{Br}\Delta H_{\text{vap2}}+x_{Cr}\Delta H_{\text{vap3}} )
\end{align}
\end{subequations}
In the separator, the algebraic equations that model the composition of the overhead vapor related to the composition in the bottom stream are as follows:
\begin{subequations}\label{paper1: cstr: alg equ}
\begin{align}
x_{Ar} &= \frac{\alpha _A x_{A3}}{\alpha _A x_{A3} + \alpha _B x_{B3} + \alpha _C x_{C3} } \\[0.3em]
x_{Br} &= \frac{\alpha _B x_{B3}}{\alpha _A x_{A3} + \alpha _B x_{B3} + \alpha _C x_{C3} } \\[0.3em]
x_{Cr} &= \frac{\alpha _C x_{C3}}{\alpha _A x_{A3} + \alpha _B x_{B3} + \alpha _C x_{C3} }
\end{align}
\end{subequations}

The first-principles model in (\ref{eq:three tank}) is utilized to build the simulator of the process. 
The nine state variables are sampled synchronously at each time instant.
More details of the chemical process can be found in \cite{liu2009distributed, zhang2013distributed}.

\subsection{Simulation setup}
The control objective is to steer state $s$ towards a pre-defined reference state $s_{ref}$. A set of reference states $\mathcal{S}_{ref}$ containing 1,375 steady-states is collected from open-loop process simulations. 
When testing the method, three reference states $s_{ref1}$, $s_{ref2}$, $s_{ref3}$ are selected from set $\mathcal{S}_{ref}$. These states respectively contain the maximum value, middle value, and minimum value of variable $x_{A1}$ in $\mathcal{S}_{ref}$.
The initial state is uniformly distributed within the interval of $[0.8s_0, 1.2s_0]$ throughout training. 
Process disturbances $w$ are generated from the Gaussian distribution $\mathcal{N}(\mathbf{0}, \sigma_{w}^2)$ and truncated within the interval of $[-b_w, b_w]$. The disturbances are added to the state derivatives obtained from simulating (\ref{eq:three tank}). Two sets of standard deviation vectors $\sigma_{w1}, \sigma_{w2}$ and bound vectors $\text{b}_{w1}, \text{b}_{w2}$ are adopted, where $\sigma_{w1}, \text{b}_{w1}$ are used in the training phase of DLAC and $\sigma_{w2}, \text{b}_{w2}$ are used in the robustness-evaluation of DLAC. 
The actions are required to be bounded by the lower bound $a_{low}$ and the higher bound $a_{high}$. 
The state-related parameters $s_{ref1}$, $s_{ref2}$, $s_{ref3}$, $s_0$, $\sigma_{w1}$, $\sigma_{w2}$, $\text{b}_{w1}$, $\text{b}_{w2}$ are listed in Table~\ref{table:state}. The action-related parameters $a_{ref1}$, $a_{ref2}$, $a_{ref3}$, $a_{low}$, $a_{high}$ are listed in Table~\ref{table:action}.

\begin{table}[tb]
  \renewcommand\arraystretch{1.5}
  \caption{The value setting related to system states}\vspace{2mm}
  \label{table:state}
  \centering
  \resizebox{\textwidth}{!}{
    \begin{tabular}{ c c c c c c c c c c }
      \toprule
         &  $x_{A1} $ & $x_{B1}$ &  $T_{1}(\text{K})$ & $x_{A2}$ & $x_{B2}$ & $T_{2}(\text{K})$ & $x_{A3}$ & $x_{B3}$ & $T_{3}(\text{K})$  \\
        \midrule
        $s_{ref1}$ &
        0.3628&	0.5961&	451.6020&	0.3768&	0.5817&	442.6132&	0.1623&	0.7490&	445.0317 \\
        $s_{ref2}$ &
        0.2063&	0.6748&	474.5915&	0.2273&	0.6546&	464.9403&	0.0794&	0.7035&	470.7400 \\
        $s_{ref3}$ &
        0.0496&	0.4003&	533.1381&	0.0686&	0.3943&	525.2897&	0.0155&	0.2858&	531.8112 \\
        $s_{0}$ &
        0.1763& 0.6731& 480.3165& 0.1965& 0.6536& 472.7863& 0.0651& 0.6703& 474.8877 \\
        $\sigma_{w1}$ &
        0.01 & 0.01 & 0.5      & 0.01 & 0.01 & 0.5      & 0.01 & 0.01 & 0.5 \\
        $\sigma_{w2}$ &
        1 & 1 & 50      & 1 & 1 & 50      & 1 & 1 & 50 \\
        $\text{b}_{w1}$ &
        5 & 5 & 5      & 5 & 5 & 5      & 5 & 5 & 5 \\
        $\text{b}_{w2}$ &
        500 & 500 & 500      & 500 & 500 & 500      & 500 & 500 & 500 \\
        \bottomrule
    \end{tabular}
    }
\end{table}

Data preprocessing is utilized. The state variables representing the mass fractions of $A$ and $B$ in the three vessels are scaled from $[0, 1]$ to $[-1, 1]$. Each state related to the temperatures is normalized using the mean and standard deviation of the respective state values in the reference state set $\mathcal{S}_{ref}$.
This data processing step aims to enhance the stability and performance of the learning process.

\begin{small}    
\begin{table}[tb]
  \renewcommand\arraystretch{1.5}
  \caption{The value setting related to system actions}\vspace{2mm}
  \label{table:action}
  \centering
    \begin{tabular}{ c c c c }
        \toprule
        Input &  $Q_1(\text{kJ/h})$ & $Q_2(\text{kJ/h})$ &  $Q_3(\text{kJ/h})$\\
        \midrule
        $a_{ref1}$ &
        3.982$\times 10^6$ & 1.114$\times 10^6$ & 1.059$\times 10^6$   \\
        $a_{ref2}$ &
        2.062$\times 10^6$ & 9.247$\times 10^5$ & 3.823$\times 10^6$   \\
        $a_{ref3}$ &
        3.829$\times 10^6$ & 4.129$\times 10^5$ & 4.044$\times 10^6$   \\
        $a_{low}$ &
        6.496$\times 10^5$ & 2.240$\times 10^5$ & 6.496$\times 10^5$   \\
        $a_{high}$ &
        4.872$\times 10^6$ & 1.680$\times 10^6$ & 4.872$\times 10^6$   \\
        \bottomrule
    \end{tabular}
\end{table}
\end{small}

The chemical process is decomposed into three subsystems, each of which accounts for one vessel of the process, as depicted in Figure~\ref{fig:three tank}. Accordingly, the subsystem states $s^1, s^2, s^3$ and subsystem actions $a^1, a^2, a^3$ for the three subsystems are:
\begin{align}
    s^1 = [x_{A1}, x_{B1}, T_1]^\top, & \quad a^1 = Q_1 \\
    s^2 = [x_{A2}, x_{B2}, T_2]^\top, & \quad a^2 = Q_2 \\
    s^3 = [x_{A3}, x_{B3}, T_3]^\top, & \quad a^3 = Q_3
\end{align}

\subsection{Evaluation}

The proposed DLAC is evaluated considering the following aspects:\vspace{-1mm}
\begin{itemize}
    \item Convergence: whether the training algorithm can converge with different initial parameters;\vspace{-1mm}
    \item Stability: whether the DLAC scheme can steer the state trajectory towards the reference;\vspace{-1mm}
    \item Performance: whether the DLAC scheme tends to minimize the cumulative stage costs;\vspace{-1mm}
    \item Robustness: whether the DLAC scheme is robust against unseen process disturbances;\vspace{-1mm}
    \item Tracking capability: whether the DLAC scheme can generalize to track multiple references.
\end{itemize}
The hyperparameters of the DLAC are listed in Table~\ref{table:param}. The values of these hyperparameters, including threshold, learning rates, and discount factor, are selected through trial-and-error analysis to achieve a good trade-off between training loss and training convergence speed.
We compare the DLAC scheme with open-loop control and nonlinear MPC (NMPC) based on the first-principles model (\ref{eq:three tank}), with the NMPC scheme developed based on \cite{han2013nonlinear}. Open-loop control is conducted using the steady-state inputs corresponding to the reference states.

\subsubsection{Convergence and performance}

\begin{table}[tb]
  \renewcommand\arraystretch{1.5}
  \caption{Hyperparameters of DLAC for the chemical process}\vspace{2mm}
  \label{table:param}
  \centering
  \resizebox{\textwidth}{!}{
    \begin{tabular}{ c c c c}
        \toprule
        Hyperparameters &  Value & Hyperparameters &  Value\\
        \midrule
        Sampling time $\Delta_t$ & 0.005 h & Episode length $n_{steps}$ & 500 \\
        Replay buffer size $n_{data\_max}$ & 4000 & Maximal episodes number $n_{eps\_max}$& 8,000 \\
        Batch size $n_{batch}$ & 256 & Training interval $n_{eps\_interval}$ & 50 \\
        Update numbers per training $n_{update}$ & 1000 & Training start episode $n_{eps\_min}$ & the 400th episode\\
        Lyapunov coefficient $\alpha_3$   & 0.5 & Threshold $\mathcal{E}$ & -1\\
        Actor learning rate & 0.0001 & Lagrange multipliers learning rate & 0.0001\\
        Critic learning rate & 0.0001 & Soft update coefficient & 5$\times 10^{-6}$\\
        Discount factor $\gamma$ & 0.95 &&\\
        \bottomrule
    \end{tabular}
    }
\end{table}

In this section, we examine the convergence of DLAC. The accumulated costs and critic losses during training are used to validate the convergence of the DLAC. During training, controllers are evaluated based on their ability to track 11 reference states that are uniformly distributed within the reference set $\s{S}_{ref}$. Two metrics are adopted for the evaluation of controllers: 
\begin{itemize}
    \item the maximum steady-state tracking error for temperatures $T_1, T_2, T_3$;
    \item the maximum steady-state tracking error for mass fractions of reactants.
\end{itemize}
\begin{figure}
    \centering
    \subfigure[Trajectory of accumulated costs during training]{
    \label{fig:costs}
    \includegraphics[width=0.48\textwidth]{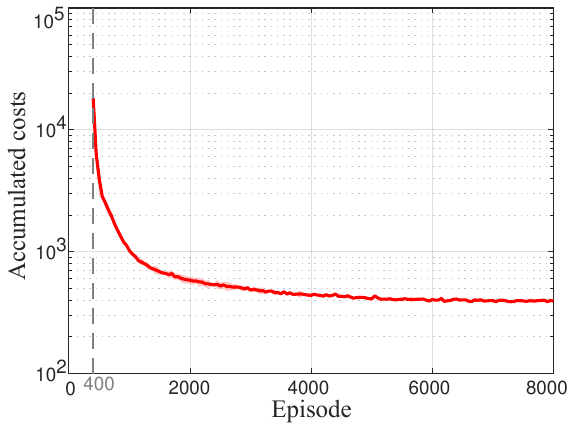}}
    \subfigure[Trajectory of the critic's loss during training]{
    \label{fig:loss}
    \includegraphics[width=0.49\textwidth]{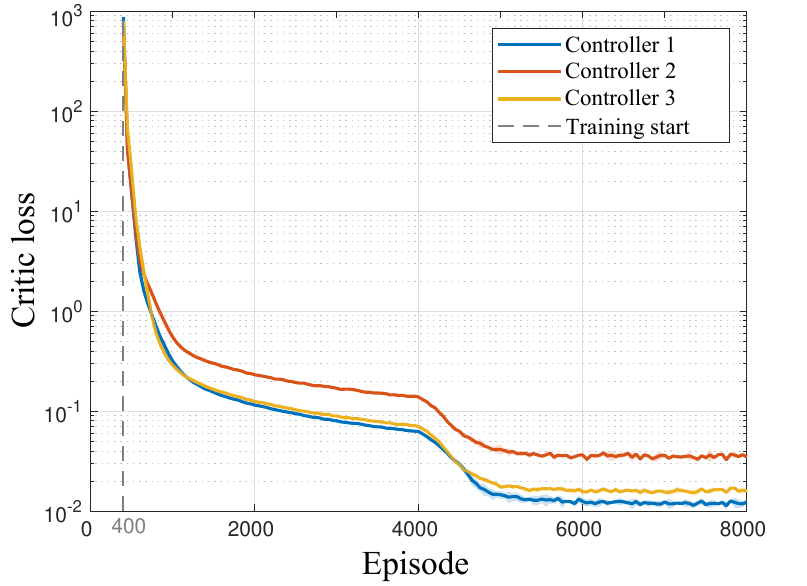}}
    \caption{Evaluation of convergence over episodes during training. Figure~\ref{fig:costs} displays the accumulated costs of DLAC, while Figure \ref{fig:loss} shows the critic losses. DLAC was trained five times with random initial states. The gray dashed line marks the start of training, the solid lines represent the mean values, and the shaded areas indicate the variance across the five training trials.}
\end{figure}

\begin{figure}[t]
    \centering
    \subfigure[Temperature tracking error]{
        \label{fig:test2}
        \includegraphics[width=0.485\textwidth]{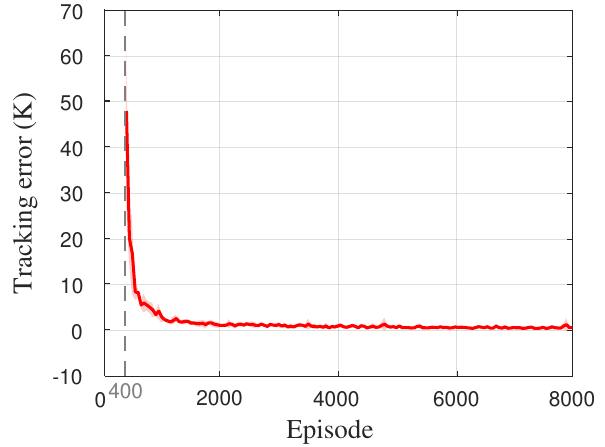}}
    \subfigure[Mass fraction tracking error]{
        \label{fig:test3}
        \includegraphics[width=0.485\textwidth]{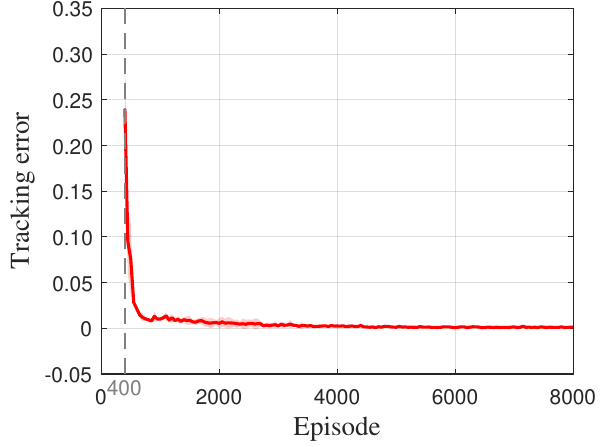}}
    \caption{Evaluation of the controller tracking performance during training. Figure~\ref{fig:test2} shows the maximum static state tracking error for the variables representing temperatures; Figure~\ref{fig:test3} shows the maximum static state tracking error for the variables representing mass fractions. DLAC is trained 5 times with random initial states. The solid lines represent the mean values, and the shaded areas indicate the variance across the 5 training trials.}
    \label{fig:test}
\end{figure}

\begin{figure}
    \centering
    \includegraphics[width=1\textwidth]{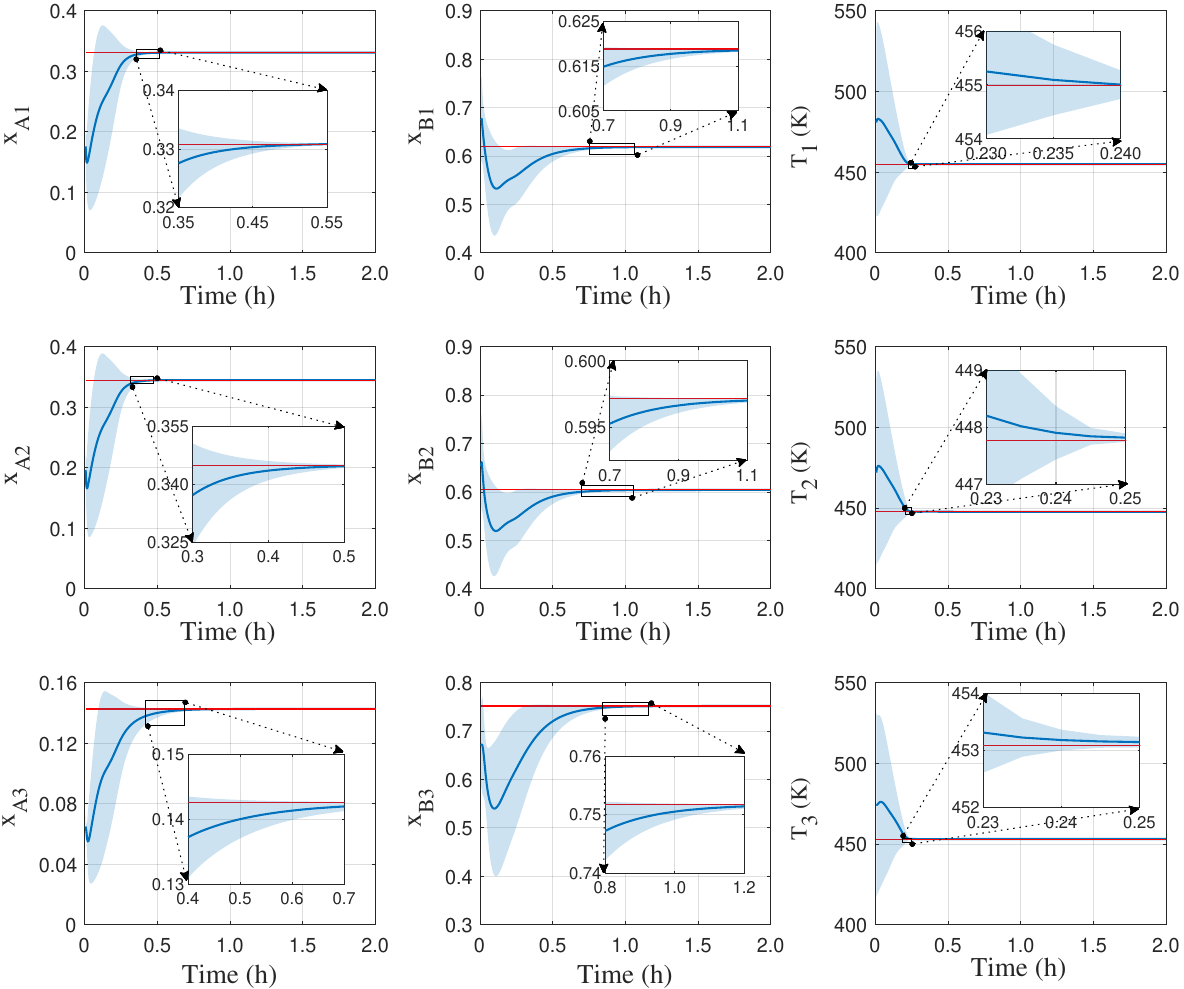}
    \caption{
    State trajectories under the process disturbances with standard deviation $\sigma_{w1}$, where the red lines represent a reference state randomly selected from the reference set $\mathcal{S}_{ref}$, the solid blue lines depict the mean of the 100 simulated trajectories with random initial states, the light blue shaded areas represent one standard deviation from the mean, and the magnified sections within each subplot provide a closer look at the critical points.
    }
    \label{fig:small_noise}
\end{figure}

The trajectory of accumulated costs during the training process is presented in Figure~\ref{fig:costs}, with the vertical axis shown on a logarithmic scale. From the 400th episode, the accumulated costs decrease rapidly over the next 1600 episodes. In the final 2000 episodes, the accumulated costs remain almost constant, even on the logarithmic scale, indicating the convergence of DLAC across random initial parameters.
Figure~\ref{fig:loss} depicts the critic losses on a logarithmic scale along these episodes for the three controllers. The training losses decrease rapidly during the initial phases and stabilize in the latter part of the training process. The shaded variance region remains relatively tight throughout the training phase, which indicates consistent performance across different state initializations.
The evaluations for the two types of steady-state error are demonstrated in Figure~\ref{fig:test}. Similar to the accumulated costs and critic losses, the rapidly declining results of these two metrics indicate good convergence of the proposed algorithm.

\begin{figure}
    \centering
    \includegraphics[width=1\textwidth]{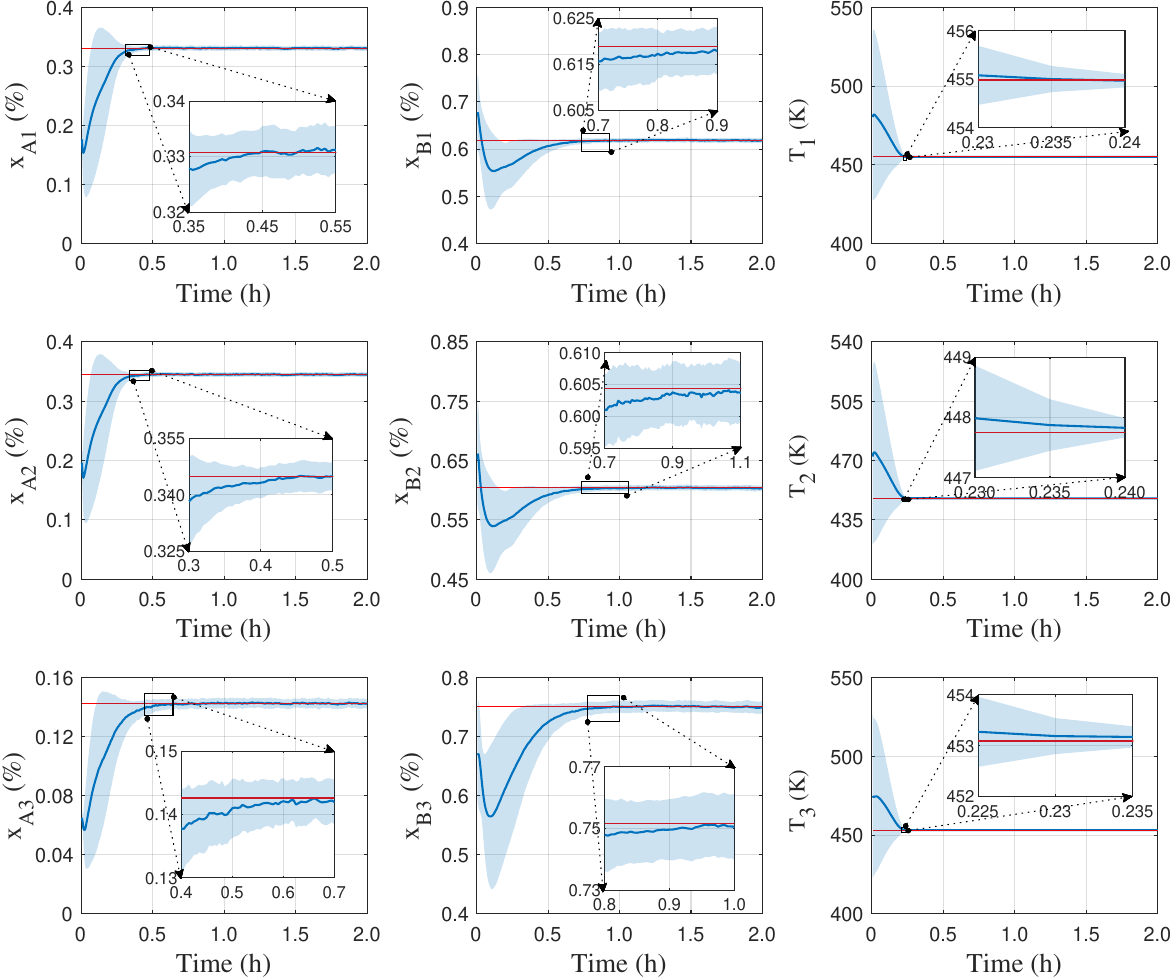}
    \caption{
    State trajectories under larger process disturbances with standard deviation $\sigma_{w2}$ are illustrated, where the red lines represent a reference state randomly selected from the reference set $\mathcal{S}_{ref}$. The solid blue lines depict the mean of 100 simulated trajectories with random initial states. The light blue shaded areas indicate one standard deviation from the mean. Magnified sections within each subplot provide a closer look at critical points.}
    \label{fig:large_noise}
\end{figure}

\begin{figure}
    \centering
    \includegraphics[width=1\textwidth]{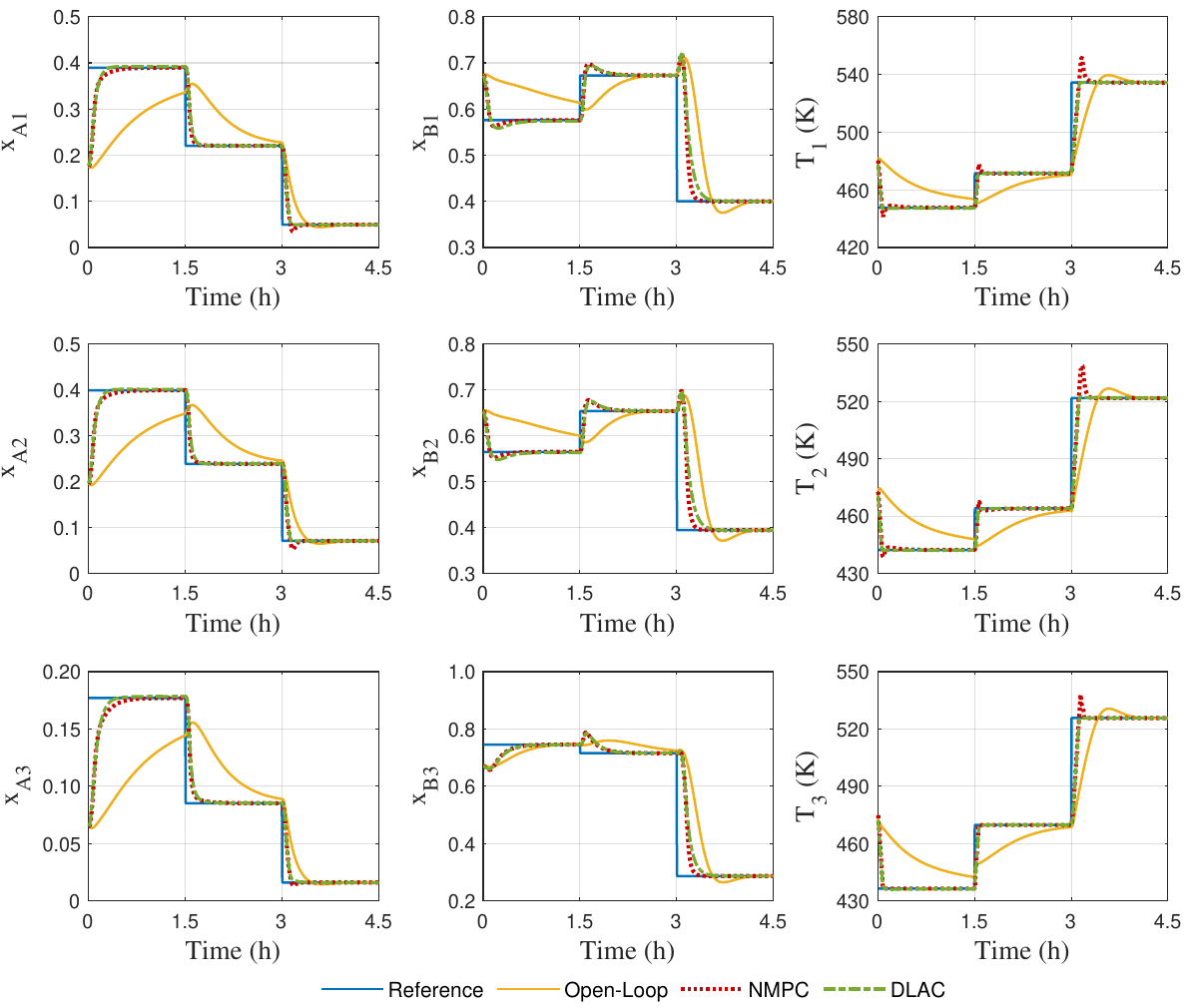}
    \caption{Comparison of dynamic tracking performance between DLAC and baselines under process disturbances with standard deviation $\sigma_{w1}$ to follow piecewise continuous reference states $s_{ref1}$, $s_{ref2}$, and $s_{ref3}$. The blue dashed line is the reference trajectory; the yellow solid line represents the state trajectory under open-loop control; the red solid line depicts the state trajectory based on NMPC; and the green solid line shows the state trajectory based on DLAC.}
    \label{fig:compare_small_noise}
\end{figure}

\begin{figure}
    \centering
    \includegraphics[width=1\textwidth]{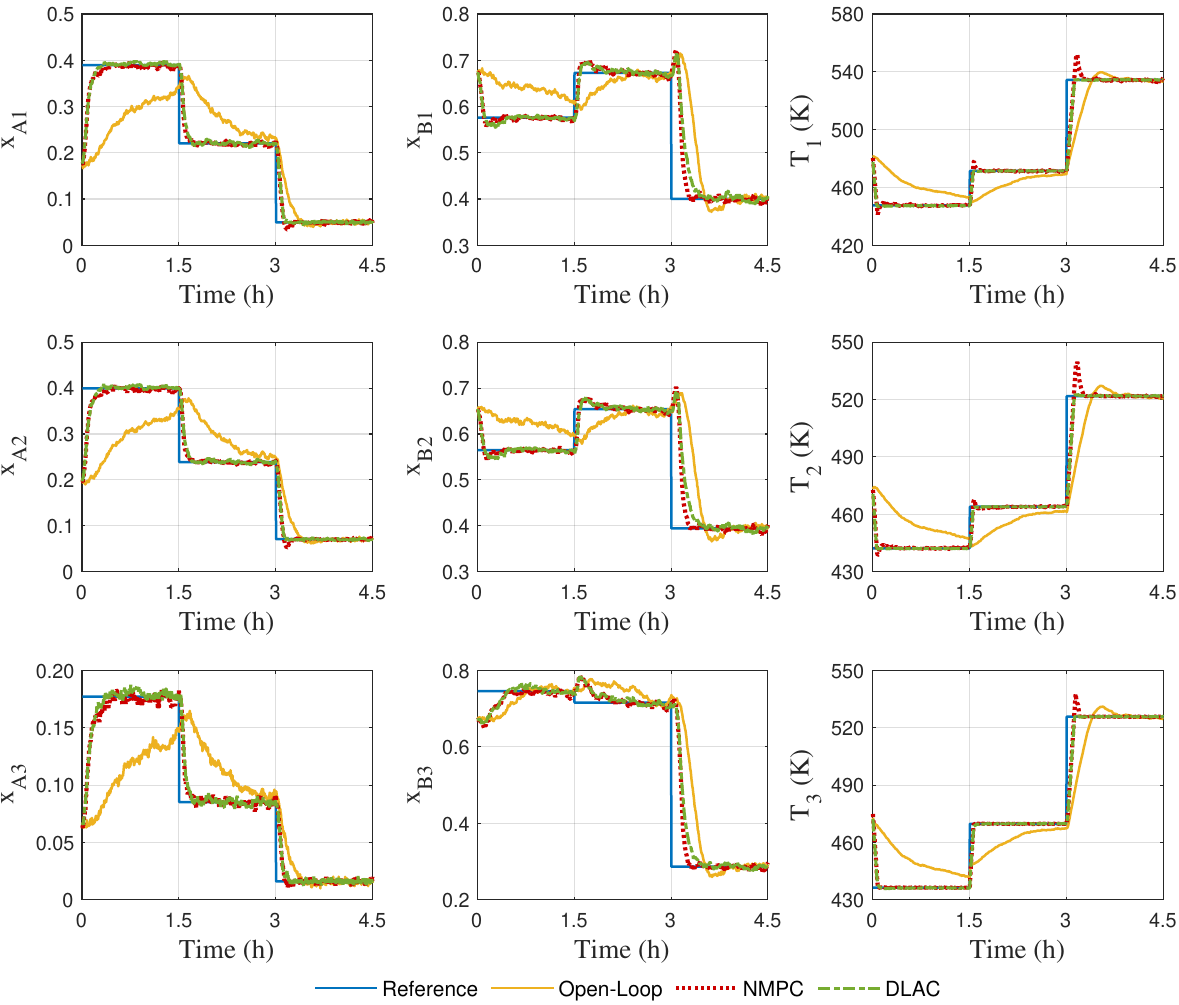}
    \caption{
    Comparison of dynamic tracking performance between DLAC and baselines under larger process disturbances with standard deviation $\sigma_{w2}$ to follow piecewise continuous reference states $s_{ref1}$, $s_{ref2}$, and $s_{ref3}$. The blue dashed line represents the reference trajectory; the yellow solid line indicates the state trajectory under open-loop control; the red solid line depicts NMPC's state trajectory; and the green solid line shows DLAC's state trajectory.
    }
    \label{fig:compare_large_noise}
\end{figure}

\subsubsection{Stability}

We evaluate the tracking performance of the trained controllers using 100 different random initial states.
Figure~\ref{fig:small_noise} shows the state trajectories of the closed-loop system and the reference states.
The red lines represent a reference state randomly selected from the reference set $\mathcal{S}_{ref}$;  
the solid blue lines depict the mean of these 100 simulated trajectories,
while the light blue shaded areas represent one standard deviation from the mean. 
The system states converge progressively towards the desired reference within approximately 1 hour from different initial states and remain close to the reference thereafter.

\subsubsection{Robustness}

In this section, we evaluate the robustness of DLAC under larger process disturbances with standard deviation $\sigma_{w2}$, which is 100 times that of the disturbances with standard deviation $\sigma_{w1}$. Figure~\ref{fig:large_noise} presents the updated tracking outcomes under the new disturbances. 

As shown in Figure~\ref{fig:large_noise}, the state trajectories become noisier as the magnitude of the disturbances increases. Despite the increased noise levels, the mean trajectories (depicted by blue lines) generally converge towards the reference lines across all dimensions. This indicates that the control system effectively compensates for these perturbations to track the desired states.

\subsubsection{Tracking of multiple references}
In this section, we consider changes in the reference states and evaluate the tracking performance of the DLAC scheme accordingly. 
Specifically, the reference state changes from $s_{ref1}$ to $s_{ref2}$ and then to $s_{ref3}$, with dwelling intervals of 1.5 hours. 
Figure~\ref{fig:compare_small_noise} and Figure~\ref{fig:compare_large_noise} present the tracking performance of the DLAC scheme tracking piecewise continuous reference states, subject to process disturbances of standard deviation $\sigma_{w1}$ and larger process disturbances with standard deviation $\sigma_{w2}$, respectively. 
Both NMPC and DLAC outperform open-loop control. Compared to NMPC, the proposed DLAC provides a shorter settling time with moderate overshoot, particularly for $T_1$, $T_2$, and $T_3$. These results demonstrate that the proposed method is able to provide good control performance without requiring a first-principles dynamic model of the process, which is one of the advantages of the proposed method.

\section{Concluding remarks}
\label{sec:conclusion}
An RL-based distributed control method, with a paradigm of distributed training and decentralized execution, was proposed. Data-based closed-loop stability was analyzed, and Lyapunov-based conditions were derived to guide the design of local controllers within this distributed framework.
The training phase of the local controllers requires only the exchange of scalar-valued information. Once the training phase is completed, controllers can be implemented in a decentralized manner, that is, they do not need to exchange information with neighboring controllers during online control implementation.
Simulations on a reactor-separator process were conducted to illustrate the proposed method. The method demonstrated good performance in terms of training convergence, set-point tracking performance, and robustness against larger disturbances.  The results also confirmed that the trained controllers can track multi-level references without the need for {retraining} the RL-based controllers.

Future research may focus on investigating sequential communication strategies to further reduce communication overhead during the training phase of RL-based distributed algorithms.
Additionally, it would be relevant to develop and validate an optimal RL-based distributed control approach with stability guarantees, as well as analyze its convergence to centralized RL. Furthermore, exploring the application of partially observable Markov decision processes (POMDPs) \cite{lee2021multi, liu2022sample} could provide insights into scenarios where only partial states can be measured online.

\section*{Acknowledgment}
This research is supported by the Ministry of Education, Singapore, under its Academic Research Fund Tier 1 (RG63/22), jointly with the Maritime Energy and Sustainable Development Centre of Excellence. This research is also supported by the National Research Foundation, Singapore, and PUB, Singapore’s National Water Agency under its RIE2025 Urban Solutions and Sustainability (USS) (Water) Centre of Excellence (CoE) Programme, awarded to Nanyang Environment \& Water Research Institute (NEWRI), Nanyang Technological University, Singapore (NTU). 


\begin{thebibliography}{10}



\bibitem{daoutidis2016JPC}
P. Daoutidis, M. Zachar, and S. S. Jogwar.
\newblock Sustainability and process control: A survey and perspective.
\newblock {\em Journal of Process Control}, 44:184--206, 2016.

\bibitem{CACE2023}
W. Tang, P. Carrette, Y. Cai, J. M. Williamson, and P. Daoutidis.
\newblock Automatic decomposition of large-scale industrial processes for distributed MPC on the Shell–Yokogawa Platform for Advanced Control and Estimation (PACE).
\newblock {\em Computers \& Chemical Engineering}, 178:108382, 2019.

\bibitem{yin2017automatica}
X. Yin and J. Liu.
\newblock Distributed moving horizon state estimation of two-time-scale nonlinear systems.
\newblock {\em Automatica}, 79:152--161, 2017.

\bibitem{TCST2019}
X. Yin, J. Zeng, and J. Liu.
\newblock Forming distributed state estimation network from decentralized estimators.
\newblock {\em IEEE Transactions on Control Systems Technology}, 27(6):2430--2443, 2019.

\bibitem{burg2016large}
J. M. Burg, C. B. Cooper, Z. Ye, B. R. Reed, E. A. Moreb, and M. D. Lynch.
\newblock Large-scale bioprocess competitiveness: the potential of dynamic metabolic control in two-stage fermentations.
\newblock {\em Current Opinion in Chemical Engineering}, 14:121--136, 2016.


\bibitem{daoutidis2018cace}
P. Daoutidis, W. Tang, and S. S. Jogwar.
\newblock Decomposing complex plants for distributed control: Perspectives from network theory.
\newblock {\em Computers \& Chemical Engineering}, 114:43--51, 2018.

\bibitem{christofides2013distributed}
P. D. Christofides, R. Scattolini, D. Mu{\~n}oz {de la Pe\~na}, and J. Liu.
\newblock Distributed model predictive control: A tutorial review and future research directions.
\newblock {\em Computers \& Chemical Engineering}, 51:21--41, 2013.

\bibitem{aiche2019}
X. Yin and J. Liu.
\newblock Subsystem decomposition of process networks for simultaneous distributed state estimation and control.
\newblock {\em AIChE Journal}, 65(3):904--914, 2019.

\bibitem{stewart2010cooperative}
B. T. Stewart, A. N. Venkat, J. B. Rawlings, S. J. Wright, and G. Pannocchia.
\newblock Cooperative distributed model predictive control.
\newblock {\em Systems \& Control Letters}, 59(8):460--469, 2010.

\bibitem{liu2011iterative}
J. Liu, X. Chen, D. Mu{\~n}oz de la Peña, and P. D. Christofides.
\newblock Iterative distributed model predictive control of nonlinear systems: Handling asynchronous, delayed measurements.
\newblock {\em IEEE Transactions on Automatic Control}, 57(2):528--534, 2011.


\bibitem{li2023aiche}
X. Li, A. W. K. Law, and X. Yin.
\newblock Partition-based distributed extended Kalman filter for large-scale nonlinear processes with application to chemical and wastewater treatment processes.
\newblock {\em AIChE Journal}, 69(12):e18229, 2023.

\bibitem{yin2021}
X. Yin, Z. Li, I. V. Kolmanovsky.
\newblock Distributed state estimation for linear systems with application to full-car active suspension systems.
\newblock {\em IEEE Transactions on Industrial Electronics}, 68(2):1615-1625, 2021.


{{\bibitem{sutton2018reinforcement}
R. S. Sutton.
\newblock Reinforcement learning: An introduction.
\newblock {\em MIT press}, 2018.}}

{{\bibitem{espeholt2018impala}
L. Espeholt, H. Soyer, R. Munos, K. Simonyan, V. Mnih, T. Ward, Y. Doron, V. Firoiu, T. Harley, I. Dunning, S. Legg, and K. Kavukcuoglu.
\newblock IMPALA: Scalable distributed deep-RL with importance weighted actor-learner architectures.
\newblock {\em International Conference on Machine Learning}, 1407-1416, 2018, Stockholm, Sweden.}}

\bibitem{kumar2012model}
A. S. Kumar and Z. Ahmad.
\newblock Model predictive control (MPC) and its current issues in chemical engineering.
\newblock {\em Chemical Engineering Communications}, 199(4):472--511, 2012.

\bibitem{rawlings2017model}
J. B. Rawlings, D. Q. Mayne, and M. Diehl.
\newblock Model predictive control: Theory, computation, and design.
\newblock {\em Nob Hill Publishing}, 2017.


\bibitem{liu2010distributed}
J. Liu, D. Mu{\~n}oz de la Pe{\~n}a, and P. D. Christofides.
\newblock Distributed model predictive control of nonlinear systems subject to asynchronous and delayed measurements.
\newblock {\em Automatica}, 46(1):52-61, 2010.

\bibitem{magni2006stabilizing}
L. Magni and R. Scattolini.
\newblock Stabilizing decentralized model predictive control of nonlinear systems.
\newblock {\em Automatica}, 42(7):1231--1236, 2006.

{{\bibitem{liu2010sequential}
J. Liu, X. Chen, D. Mu{\~n}oz de la Pe{\~n}a and P. D. Christofides.
\newblock Sequential and iterative architectures for distributed model predictive control of nonlinear process systems.
\newblock {\em AIChE Journal}, 
56(8):2137-2149, 2010.}}

{{
\bibitem{heidarinejad2013distributed}
M. Heidarinejad, J. Liu and P. D. Christofides.
\newblock Distributed model predictive control of switched nonlinear systems with scheduled mode transitions.
\newblock {\em AIChE Journal}, 
59(3):860-871, 2013.}}

\bibitem{zhao2023feature}
T. Zhao, Y. Zheng, and Z. Wu.
\newblock Feature selection-based machine learning modeling for distributed model predictive control of nonlinear processes.
\newblock {\em Computers \& Chemical Engineering}, 169:108074, 2023.

\bibitem{kadakia2023encrypted}
Y.A. Kadakia, A. Alnajdi, F. Abdullah, and P.D. Christofides.
\newblock Encrypted decentralized model predictive control of nonlinear processes with delays.
\newblock {\em Chemical Engineering Research and Design}, 200:312-324, 2023.

\bibitem{kadakia2024encrypted}
Y.A. Kadakia, F. Abdullah, A. Alnajdi, and P. D. Christofides. 
\newblock Encrypted distributed model predictive control of nonlinear processes.
\newblock {\em  Control Engineering Practice}, 145:105874, 2024.

\bibitem{lee2018machine}
J. H. Lee, J. Shin, and M. J. Realff.
\newblock Machine learning: Overview of the recent progresses and implications for the process systems engineering field.
\newblock {\em Computers \& Chemical Engineering}, 114:111--121, 2018.

\bibitem{shin2019reinforcement}
J. Shin, T. A. Badgwell, K. H. Liu, and J. H. Lee.
\newblock Reinforcement Learning - Overview of recent progress and implications for process control.
\newblock {\em Computers \& Chemical Engineering}, 127:282--294, 2019.

\bibitem{nian2020review}
R. Nian, J. Liu and B. Huang.
\newblock A review on reinforcement learning: Introduction and applications in industrial process control.
\newblock {\em Computers \& Chemical Engineering}, 139:106886, 2020.

\bibitem{spielberg2017deep}
S. P. K. Spielberg, R. B. Gopaluni, and P. D. Loewen.
\newblock Deep reinforcement learning approaches for process control.
\newblock {\em International Symposium on Advanced Control of Industrial Processes}, 201--206, 2017, Taipei, Taiwan.

{{
\bibitem{ma2019continuous}
Y. Ma, W. Zhu, M. G. Benton, and J. Romagnoli.
\newblock Continuous control of a polymerization system with deep reinforcement learning.
\newblock {\em Journal of Process Control}, 75:40--47, 2019.
}}

\bibitem{wang2024control}
Y. Wang and Z. Wu.
\newblock Control Lyapunov-barrier function-based safe reinforcement learning for nonlinear optimal control.
\newblock {\em AIChE Journal}, 70(3):e18306, 2024.

\bibitem{hoskins1992process}
J. C. Hoskins and D. M. Himmelblau.
\newblock Process control via artificial neural networks and reinforcement learning.
\newblock {\em Computers \& Chemical Engineering}, 16(4):241--251, 1992.


\bibitem{foerster2016learning}
J. Foerster, I. A. Assael, N. De Freitas, and S. Whiteson.
\newblock Learning to communicate with deep multi-agent reinforcement learning.
\newblock {\em Advances in Neural Information Processing Systems}, 29, 2016.

\bibitem{canese2021multi}
L. Canese, G. C. Cardarilli, L. D. Nunzio, R. Fazzolari, D. Giardino, M. Re, and S. Spanò.
\newblock Multi-agent reinforcement learning: A review of challenges and applications.
\newblock {\em Applied Sciences}, 11(11):4948, 2021.

\bibitem{shaikhet1997necessary}
L. Shaikhet.
\newblock Necessary and sufficient conditions of asymptotic mean square stability for stochastic linear difference equations.
\newblock {\em Applied Mathematics Letters}, 10(3):111--115, 1997.

\bibitem{bolzern2010markov}
P. Bolzern, P. Colaneri, and G. De Nicolao.
\newblock Markov jump linear systems with switching transition rates: Mean square stability with dwell-time.
\newblock {\em Automatica}, 46(6):1081--1088, 2010.

\bibitem{han2020actor}
M. Han, L. Zhang, J. Wang, and W. Pan.
\newblock Actor-critic reinforcement learning for control with stability guarantee.
\newblock {\em IEEE Robotics and Automation Letters}, 5(4):6217--6224, 2020.

\bibitem{sutton2008convergent}
R. S. Sutton, H. Maei, and C. Szepesvári.
\newblock A convergent $ O (n) $ temporal-difference algorithm for off-policy learning with linear function approximation.
\newblock {\em Advances in Neural Information Processing Systems}, 21, 2008.

\bibitem{korda2015td}
N. Korda and P. La.
\newblock On TD(0) with function approximation: Concentration bounds and a centered variant with exponential convergence.
\newblock {\em International Conference on Machine Learning}, 626--634, 2015, Lille, France.

\bibitem{bhandari2018finite}
J. Bhandari, D. Russo, and R. Singal.
\newblock A finite time analysis of temporal difference learning with linear function approximation.
\newblock {\em Conference on Learning Theory}, 1691--1692, 2018, Stockholm, Sweden.


{
\bibitem{puterman2014markov}
M. L. Puterman.
\newblock Markov decision processes: Discrete stochastic dynamic programming.
\newblock {\em John Wiley \& Sons}, 2014.}

{
\bibitem{TH2008}
A. Gelman and D. B. Rubin.
\newblock Inference from iterative simulation using multiple sequences.
\newblock {\em Statistical Science},
7(4):457-472, 1992.
}

\bibitem{khalil2009lyapunov}
H. K. Khalil.
\newblock Lyapunov stability.
\newblock {\em Control Systems, Robotics and Automation}, 12:115, 2009.

\bibitem{raghunathan2013optimal}
A. Raghunathan and U. Vaidya.
\newblock Optimal stabilization using Lyapunov measures.
\newblock {\em IEEE Transactions on Automatic Control}, 59(5):1316--1321, 2013.

\bibitem{blackwell1963converse}
D. Blackwell and L. E. Dubins.
\newblock A converse to the dominated convergence theorem.
\newblock {\em Illinois Journal of Mathematics}, 7(3):508--514, 1963.
.

\bibitem{chow2018lyapunov}
Y. Chow, O. Nachum, E. Duenez-Guzman, and M. Ghavamzadeh.
\newblock A Lyapunov-based approach to safe reinforcement learning.
\newblock {\em Advances in Neural Information Processing Systems}, 31, 2018.

\bibitem{robbins1951stochastic}
H. Robbins and S. Monro.
\newblock A stochastic approximation method.
\newblock {\em The Annals of Mathematical Statistics}, 400--407, 1951.

\bibitem{haarnoja2018soft}
T. Haarnoja, A. Zhou, K. Hartikainen, G. Tucker, S. Ha, J. Tan, V. Kumar, H. Zhu, A. Gupta, P. Abbeel, A. Gupta, P. Abbeel, and S. Levine.
\newblock Soft actor-critic algorithms and applications.
\newblock {\em arXiv preprint arXiv:1812.05905}, 2018.


\bibitem{liu2009distributed}
J. Liu, D. Muñoz de la Peña, and P. D. Christofides.
\newblock Distributed model predictive control of nonlinear process systems.
\newblock {\em AIChE Journal}, 55(5):1171--1184, 2009.

\bibitem{zhang2013distributed}
J. Zhang and J. Liu.
\newblock Distributed moving horizon state estimation for nonlinear systems with bounded uncertainties.
\newblock {\em Journal of Process Control}, 23(9):1281--1295, 2013.

\bibitem{han2013nonlinear}
H. Han and J. Qiao.
\newblock Nonlinear model-predictive control for industrial processes: An application to wastewater treatment process.
\newblock {\em IEEE Transactions on Industrial Electronics}, 61(4):1970--1982, 2013.

\bibitem{lee2021multi}
H. R. Lee and T. Lee.
\newblock Multi-agent reinforcement learning algorithm to solve a partially-observable multi-agent problem in disaster response.
\newblock {\em European Journal of Operational Research},
291(1):296--308, 2021.

\bibitem{liu2022sample}
Q. Liu and C. Szepesv{\'a}ri, and C. Jin.
\newblock Sample-efficient reinforcement learning of partially observable Markov games.
\newblock {\em Advances in Neural Information Processing Systems},
35:18296--18308, 2022.



\end{thebibliography}

\end{document}